\RequirePackage{fix-cm}
\documentclass[smallextended]{svjour3}       
\smartqed  

\usepackage{amssymb,amsmath,amsfonts,mathrsfs}
\usepackage[utf8]{inputenc}
\usepackage[english]{babel}
\usepackage{graphicx}
\usepackage{color}
\usepackage{doi}
\usepackage{caption}
\usepackage{commath}

\DeclareMathOperator{\RR}{\mathbb{R}}

\DeclareMathOperator{\ZZ}{\mathbb{Z}}

\DeclareMathOperator{\BC}{\mathcal{B}}

\DeclareMathOperator{\VC}{\mathcal{V}}

\DeclareMathOperator{\PC}{\mathcal{P}}

\DeclareMathOperator{\SC}{\mathcal{S}}
\DeclareMathOperator{\ZC}{\mathcal{Z}}

\DeclareMathOperator{\MC}{\mathcal{M}}

\DeclareMathOperator{\FC}{\mathcal{F}}

\DeclareMathOperator{\UC}{\mathcal{U}}

\DeclareMathOperator{\JC}{\mathcal{J}}
\DeclareMathOperator{\IC}{\mathcal{I}}

\DeclareMathOperator{\BS}{\mathscr{B}}

\DeclareMathOperator{\NotBC}{\overline{\BC}}

\DeclareMathOperator{\rank}{rank}

\DeclareMathOperator{\conv}{conv\!.\!hull}

\DeclareMathOperator{\vertex}{vert}

\DeclareMathOperator{\BUnit}{\mathbf 1}
\DeclareMathOperator{\BZero}{\mathbf 0}

\DeclareMathOperator{\supp}{supp}
\DeclareMathOperator{\zeros}{zeros}

\newcommand*{\intint}[2][1]{\{#1, \dots, #2\}}

\DeclareMathOperator{\Deq}{\overset{\Delta}{=}}

\title{On a Simple Connection Between $\Delta$-modular ILP and LP, and a New Bound on the Number of Integer Vertices}

\begin{document}

\author{D.~V.~Gribanov, D.~S.~Malyshev, I.~A.~Shumilov}
%
%
\institute{D.~V.~Gribanov \at National Research University Higher School of Economics, 25/12 Bolshaja Pecherskaja Ulitsa, Nizhny Novgorod, 603155, Russian Federation\\
\email{dimitry.gribanov@gmail.com}
\and
D.~S.~Malyshev \at National Research University Higher School of Economics, 25/12 Bolshaja Pecherskaja Ulitsa, Nizhny Novgorod, 603155, Russian Federation;
\at Huawei, Intelligent systems and Data science Technology center (2012 Laboratories), 7/9 Smolenskaya Square, Moscow, 121099, Russian Federation\\
\email{dsmalyshev@rambler.ru}
\and
I.~A.~Shumilov \at Lobachevsky State University of Nizhny Novgorod, 23 Gagarina Avenue, Nizhny Novgorod, 603950, Russian Federation\\
\email{ivan.a.shumilov@gmail.com}
}

\maketitle

\begin{abstract}
    Let $A \in \ZZ^{m \times n}$, $\rank(A) = n$, $b \in \ZZ^m$, and $\PC$ be an $n$-dimensional polyhedron, induced by the system $A x \leq b$.
    
    It is a known fact that if $\FC$ is a $k$-face of $\PC$, then there exist at least $n-k$ linearly independent inequalities of the system $A x \leq b$ that become equalities on $\FC$. In other words, there exists a set of indices $\JC$, such that $|\JC| \geq n-k$, $\rank(A_{\JC}) = n-k$, and
    $$
        A_{\JC} x - b_{\JC} = \BZero,\quad \text{for any $x \in \FC$}.
    $$
    
    We show that a similar fact holds for the integer polyhedron 
    $$
    \PC_{I} = \conv\bigl(\PC \cap \ZZ^n\bigr),
    $$
    if we additionally suppose that $\PC$ is $\Delta$-modular, for some $\Delta \in \{1,2,\dots\}$. More precisely, if $\FC$ is a $k$-face of $\PC_{I}$, then there exists a set of indices $\JC$, such that $|\JC| \geq n-k$, $\rank(A_{\JC}) = n-k$, and 
    $$
        A_{\JC} x - b_{\JC} \Deq \BZero,\quad \text{for any $x \in \FC \cap \ZZ^n$},
    $$ where $x \Deq y$ means that $\|x - y\|_{\infty} < \Delta$. In other words, there exist at least $n-k$ linearly independent inequalities of the system $A x \leq b$ that almost become equalities on $\FC \cap \ZZ^n$. When we say almost, we mean that the slacks are not greater than $\Delta-1$. Using this fact, we  prove the inequality 
    $$
    |\vertex(\PC_I)| \leq 2 \cdot \binom{m}{n} \cdot \Delta^{n-1},
    $$ for the number of vertices of $\PC_I$, which is better, than the state of the art bound for $\Delta = O(n^2)$.
\end{abstract}

\section{Basic Definitions and Notations}

The maximum absolute value of entries of a matrix $A$ (also known as \emph{the matrix $\max$-norm}) is denoted by $\|A\|_{\max} = \max_{i,j} \abs{A_{i\,j}}$. The number of non-zero components of a vector $x$ is denoted by $\norm{x}_0=\abs{\{i\colon x_i \not= 0\}}$. 

Let $v \in \RR^n$. By $\supp_{\Delta}(v)$ and $\zeros_{\Delta}(v)$ we denote $\{i \colon |v_i| \geq \Delta \}$ and $\intint n \setminus \supp_{\Delta}(v)$, respectively. Denote $\supp(v) := \supp_0(v)$ and $\zeros(v) := \zeros_0(v)$. Clearly, $\norm{v}_0 = \abs{\supp(v)}$.

\begin{definition}
For a matrix $A \in \ZZ^{m \times n}$, by $$
\Delta_k(A) = \max\left\{\abs{\det (A_{\IC \JC})} \colon \IC \subseteq \intint m,\; \JC \subseteq \intint n,\; \abs{\IC} = \abs{\JC} = k\right\},
$$ we denote the maximum absolute value of determinants of all the $k \times k$ sub-matrices of $A$. Here, the symbol $A_{\IC \JC}$ denotes the sub-matrix of $A$, which is generated by all the rows with indices in $\IC$ and all the columns with indices in $\JC$. The matrix $A$ with $\Delta(A) \leq \Delta$, for some $\Delta > 0$, is called \emph{$\Delta$-modular}. Note that $\Delta_1(A) = \|A\|_{\max}$.

\end{definition}

For a matrix $B \in \RR^{m \times n}$,
$\conv(B) = \{B t\colon t \in \RR_+^{n},\, \sum_{i=1}^{n} t_i = 1  \}$ is the \emph{convex hull spanned by columns of} $B$.


\section{A simple connection between $\Delta$-modular ILP and LP}

Let $A \in \ZZ^{m \times n}$, $\rank(A) = n$, $b \in \ZZ^m$, and $\PC$ be the $n$-dimensional polyhedron, defined by the system $A x \leq b$.

Let $\FC$ be a $k$-dimensional face of $\PC$. It is a known fact from the theory of linear inequalities that there exist $n-k$ linearly independent inequalities of $A x \leq b$ that become equalities on $\FC$. More precisely, there exists a set of indices $\JC \subseteq \intint m$, such that $|\JC| \geq n - k$, $\rank(A_{\JC}) = n - k$, and 
$$
A_{\JC} x - b_{\JC} = 0,\quad \text{for $x \in \FC$},
$$ and, consequently,
$$
\abs{\supp(A x - b)} \leq m - n + k,\quad \text{for $x \in \FC$}.
$$

We are going to prove a similar fact for the polyhedron $\PC_I = \conv\bigl(\PC \cap \ZZ^n\bigr)$.

\begin{theorem}\label{deep_base_th}
Let $\FC$ be a $k$-dimensional face of $\PC_I$ and $\Delta = \Delta(A)$. Then, there exists a set of indices $\JC \subseteq \intint m$, such that $|\JC| \geq n-k$, $\rank(A_{\JC}) = n-k$, and
$$
A_{\JC} x - b_{\JC} \Deq \BZero,\quad \text{for any $x \in \FC \cap \ZZ^n$},
$$ and, consequently,
$$
\abs{\supp_{\Delta}(A x - b)} \leq m - n + k,\quad \text{for any $x \in \FC \cap \ZZ^n$}.
$$
\end{theorem}

\begin{proof}
Let consider a point $v \in \ZZ^n$, lying on a $k$-dimensional face $\FC$ of $\PC_I$, and the corresponding slacks vector $u = b - A v$. Let $\SC = \supp_{\Delta}(u)$ and $\ZC = \zeros_{\Delta}(u)$. Suppose to the contrary that $r := \rank(A_{\ZC}) < n-k$. We have

$$
\begin{pmatrix}
A_{\ZC} \\
A_{\SC}
\end{pmatrix} v + \begin{pmatrix}
u_{\ZC} \\
u_{\SC} 
\end{pmatrix} = \begin{pmatrix}
b_{\ZC} \\
b_{\SC}
\end{pmatrix}.
$$

There exists a unimodular matrix $Q \in \ZZ^{n \times n}$, such that $A_{\ZC} = \bigl(H \, \BZero\bigr) Q$, where $\bigl(H \, \BZero\bigr)$ is the HNF of $A_{\ZC}$ and $H \in \ZZ^{|\ZC| \times r}$. The zero sub-matrix of $\bigl(H\, \BZero\bigr)$ has $n - r > k$ columns. Let $y = Q v$, then
$$
\begin{pmatrix}
H & \BZero \\
C & B
\end{pmatrix} y + \begin{pmatrix}
u_{\ZC} \\
u_{\SC}
\end{pmatrix} = \begin{pmatrix}
b_{\ZC} \\
b_{\SC}
\end{pmatrix},
$$
where $\bigl(C \, B\bigr) = A_{\SC} Q^{-1}$ and $B \in \ZZ^{|\SC| \times (n-r)}$. The matrix $B$ has a full column rank $n-r$, has at list $k$ columns, and it is $\Delta$-modular. 
Consider the last $|\SC|$ equalities of the previous system. They can be written out as follows:
$$
B z + u_{\SC} = b_{\SC} -  C y_{\intint r},
$$ where $z = y_{\intint[(r+1)]{n}}$ is composed of last $n-r$ components of $y$.

From the definition of $\SC$, it follows that $(u_{\SC})_i \geq \Delta$, for any $i \in \intint{|\SC|}$. W.l.o.g., assume that $B$ is reduced to the HNF. Hence, due to Gribanov, Malyshev et al. \cite[Lemma~1]{FPT_Grib}, $\|B\|_{\max} \leq \Delta$. Let $h_1, h_2, \dots, h_{n-r}$ be the columns of $B$. Consequently, any point of the type $z \pm e_1 \pm e_2 \pm \dots \pm e_{(n-r)}$ with its corresponding slack vector $u_S \pm h_1 \pm h_2 \pm \dots \pm h_{(n-r)}$ is feasible. Since $n-r > k$, the last fact contradicts to the fact that the original point $v$ lies on the $k$-dimensional face of $\PC_I$.
\end{proof}

\section{The number of integer vertices}

\subsection{Related works}

Let us survey some remarkable results on the value of $|\vertex(\PC_I)|$. Let $\xi(n,m)$ denote the maximum number of vertices in $n$-dimensional polyhedron with $m$ facets. Due to the seminal paper \cite{MaxFacesTh} of P.~McMullen, the value of $\xi(n,m)$ attains its maximum on the class of polytopes that are dual to cyclic polytopes with $m$ vertices.

Due to the book of B.~Gr\"unbaum \cite[Section~4.7]{Grunbaum}, we have:
\begin{equation*}\label{cyclic_poly_faces_num}
    \xi(n,m) = \begin{cases}
    \frac{m}{m-s} \binom{m-s}{s},\text{ for }n = 2s\\
    2\binom{m-s-1}{s},\text{ for }n = 2s+1\\
    \end{cases} = O\left(\frac{m}{n}\right)^{n/2}.
\end{equation*}

Due to Veselov \& Chirkov \cite{IntVertEstimates_VesChir} (see also \cite{IntVerticesSurveyPart1,IntVerticesSurveyPart2}), we have:
\begin{multline}\label{int_vert_bound_delta_ext}
    \abs{\vertex(\PC_I)} \leq (n+1)^{n+1} \cdot n! \cdot \xi(n,m) \cdot \log_2^{n-1}(2 \sqrt{n+1} \cdot \Delta_{ext}) = \\
    = m^{\frac{n}{2}} \cdot O(n)^{\frac{3}{2}n+1.5} \cdot \log^{n-1}(n \cdot \Delta_{ext}),
\end{multline}
Here $\Delta_{ext} = \Delta(\bigl(A\,b\bigr))$ is the maximal absolute value of $n \times n$ sub-determinants of the augmented matrix $\bigl(A\,b\bigr)$.

Let $\phi$ be the bit-encoding length of $A x \leq b$. Due to the book of A.~Schrijver \cite[Chapter~3.2, Theorem~3.2]{Schrijver}, we have $\Delta_{ext} \leq 2^\phi$. In notation with $\phi$, the last bound \eqref{int_vert_bound_delta_ext} becomes
$$
m^{\frac{n}{2}} \cdot O(n)^{\frac{3}{2}n+1.5} \cdot (\phi + \log n)^{n-1},
$$
which outperforms a more known bound
\begin{equation}\label{int_vert_cook}
    m \cdot \binom{m-1}{n-1} \cdot (5 n^2 \cdot \phi + 1)^{n-1} = m^n \cdot \Omega(n)^{n-1} \cdot \phi^{n-1},
\end{equation} due to Cook, Hartmann et al. \cite{IntVert_Cook}, because $m \geq n$ and \eqref{int_vert_bound_delta_ext} depends on $m$ as $m^{n/2}$.

Due to Veselov \& Chirkov \cite{IntVerticesSurveyPart2}, the previous inequality \eqref{int_vert_bound_delta_ext} could be combined with the sensitivity result of Cook, Gerards et~al. \cite{Sensitivity_Tardos} to construct a bound that depends on $\Delta$ instead of $\Delta_{ext}$:
\begin{multline}\label{int_vert_bound_delta}
\abs{\vertex(\PC_I)} \leq (n+1)^{n+1} \cdot n! \cdot \xi(n,m) \cdot \xi(n,2m) \cdot \log_2^{n-1}(2 \cdot (n+1)^{2.5} \cdot \Delta^2) = \\
= m^n \cdot O(n)^{n+1.5} \cdot \log^{n-1}(n \cdot \Delta),
\end{multline}
which again is better than the bound \eqref{int_vert_cook} due to Cook, Hartmann et al., because \eqref{int_vert_bound_delta} depends only from the bit-encoding length of $A$, while \eqref{int_vert_cook} depends on the length of both $A$ and $b$.

In our work, we will prove the bound:
\begin{equation}\label{main_bound}
\abs{\vertex(\PC_I)} \leq 2 \cdot \binom{m}{n} \cdot \Delta^{n-1},
\end{equation} which outperforms the state of the art bound \eqref{int_vert_bound_delta} for $\Delta = O(n^2)$. Additionally, using results of the papers \cite{columns_delta_modular} and  \cite{ModularDiffColumns}, due to Averkov \& Schymura and Lee, Paat \& Stallknecht, we give two bounds that are independent on $m$: 
\begin{gather*}
    \abs{\vertex(\PC_I)}  = O(n)^n \cdot \Delta^{3n-1},\\
    \abs{\vertex(\PC_I)}  = O(n)^{3n} \cdot \Delta^{2n - 1}.
\end{gather*}

Some special cases deserve an additional consideration. Assume that $\PC$ is defined by a system in the standard form
$$
\begin{cases}
    A x = b\\
    x \in \RR^n_{\geq 0},
\end{cases}
$$ where $A \in \ZZ^{k \times n}$, $b \in \ZZ^k$ and $\rank(A) = k$. It is natural to call the value of $k$ as the \emph{co-dimension} of $A$ or $\PC$. The next bounds on $|\vertex(\PC_I)|$ assume that the co-dimension of $\PC$ is bounded.

Let $\Delta_1 = \Delta_1(A)$, then, due to Aliev, De~Loera et al. \cite{SupportIPSolutions}:
\begin{equation}\label{vertex_bound_supportIP}
    \abs{\vertex(\PC_I)} = (n \cdot k \cdot \Delta_1)^{O(k^2 \cdot \log(\sqrt{k} \cdot \Delta_1))}. 
\end{equation}

It is possible to improve the last bound. Let $s = \max\bigl\{\norm{v}_0 \colon v \in \vertex(\PC_I)\bigr\}$ be the \emph{sparsity parameter} of $\PC_I$. Due to Berndt, Jansen \& Klein \cite{NewBoundsForFixedM}, we have:
\begin{equation}\label{vertex_bound_NewBounds}
    \abs{\vertex(\PC_I)} = n^{k + s} \cdot s \cdot O(k)^{s-k} \cdot \log^{s}(k \cdot \Delta_1).
\end{equation}

The following improvement of \eqref{vertex_bound_NewBounds} was proposed in the work \cite{OnCanonicalProblems_Grib}, due to Gribanov, Shumilov et al.:
\begin{equation}\label{vertex_bound_OnCanonical}
    \abs{\vertex(\PC_I)} = n^s \cdot O(s)^{s+1} \cdot O(k)^{s-1} \cdot \log^{s-1}(k \cdot \Delta_1).
\end{equation}

Since $s = O\bigl( k \cdot \log(k \Delta_1) \bigr)$, due to Aliev, De~Loera et al. \cite{SupportIPSolutions}, we substitute $s$ to both bounds \eqref{vertex_bound_NewBounds} and \eqref{vertex_bound_OnCanonical}, and get:
\begin{equation*}
    \abs{\vertex(\PC_I)} = \bigl( n \cdot k \cdot \log(k \Delta_1) \bigr)^{O\bigl(k \cdot \log(k \Delta_1)\bigr)},
\end{equation*} which outperforms the bound \eqref{vertex_bound_supportIP}, due to \cite{SupportIPSolutions}. The last equality was proposed in Berndt, Jansen \& Klein \cite{NewBoundsForFixedM}. 

Due to Gribanov, Shumilov et al. \cite{OnCanonicalProblems_Grib}, it holds $s = O(k + \log(\Delta))$, where $\Delta = \Delta(A)$. Consequently, the bound \eqref{vertex_bound_OnCanonical} could be used to estimate $\abs{\vertex(\PC_I)}$ with respect to the $\Delta$ parameter instead of $\Delta_1$:
\begin{equation}\label{vertex_bound_Delta}
    \abs{\vertex(\PC_I)} = \bigl(n \cdot k \cdot \log(\Delta) \bigr)^{O\bigl(k + \log(\Delta)\bigr)}.
\end{equation}

Note that, due to \cite{OnCanonicalProblems_Grib}, the bounds \eqref{vertex_bound_OnCanonical} and \eqref{vertex_bound_Delta} can be used to work with the systems $A x \leq b$ having $m = n + k$ rows. Therefore, for the case when $\PC$ is defined by $A x \leq b$, it is also convenient to call $k$ as the co-dimension of $\PC$.

We summarise all the bounds in the following tables:
\begin{table}[ht]
    \centering
    \begin{tabular}{||c|c||}
    \hline
        $m^n \cdot O(n)^{n-1} \cdot \phi^{n-1}$ & due to Cook, Hartmann et al. \cite{IntVert_Cook}  \\
        \hline
         $m^{\frac{n}{2}} \cdot O(n)^{\frac{3}{2}n+1.5} \cdot \log^{n-1}(n \cdot \Delta_{ext}) =$ & \\
         $= m^{\frac{n}{2}} \cdot O(n)^{\frac{3}{2}n+1.5} \cdot (\phi + \log n)^{n-1} $ & due to Veselov \& Chirkov \cite{IntVertEstimates_VesChir} \\
         \hline
         $m^n \cdot O(n)^{n+1.5} \cdot \log^{n-1}(n \cdot \Delta)$ & due to Veselov \& Chirkov \cite{IntVerticesSurveyPart2} \\
         
         \hline
         $2 \cdot \binom{m}{n} \cdot \Delta^{n-1} =$ & \\
         $m^n \cdot \Omega(n)^{-n} \cdot \Delta^{n-1}$ & {\color{red} this work} \\
         \hline
          $O(n)^n \cdot \Delta^{3n-1}$ & {\color{red} this work} \\
          \hline
          $O(n)^{3 n} \cdot \Delta^{2n-1}$ & {\color{red} this work} \\
         \hline
    \end{tabular}
    
    \captionsetup{justification=centering}
    \caption{General bounds on $\abs{\vertex(\PC_I)}$}
\end{table}

\begin{table}[ht]
    \centering
    \begin{tabular}{||c|c||}
    \hline
    
        $(n \cdot k \cdot \Delta_1)^{O(k^2 \cdot \log(\sqrt{k} \cdot \Delta_1))}$ & due to Aliev, De~Loera et al. \cite{SupportIPSolutions} \\
        
        \hline
        
        $n^{k + s} \cdot s \cdot O(k)^{s-k} \cdot \log^{s}(k \cdot \Delta_1) = $ & \\
        $= \bigl( n \cdot k \cdot \log(k \Delta_1) \bigr)^{O\bigl(k \cdot \log(k \Delta_1)\bigr)}$ & due to Berndt, Jansen \& Klein \cite{NewBoundsForFixedM} \\
        
        \hline
        
        $n^s \cdot O(s)^{s+1} \cdot O(k)^{s-1} \cdot \log^{s-1}(k \cdot \Delta_1) = $ & due to Gribanov, Shumilov et al. \cite{OnCanonicalProblems_Grib} plus \\
        $= \bigl( n \cdot k \cdot \log(k \Delta_1) \bigr)^{O\bigl(k \cdot \log(k \Delta_1)\bigr)}$ & Aliev, De~Loera et al. \cite{SupportIPSolutions} \\
        
        \hline
        
        $\bigl(n \cdot k \cdot \log(\Delta) \bigr)^{O\bigl(k + \log(\Delta)\bigr)}$ & due to Gribanov, Shumilov et al. \cite{OnCanonicalProblems_Grib} \\
        
        \hline
    \end{tabular}
    
    \captionsetup{justification=centering}
    \caption{Bounds for $\abs{\vertex(\PC_I)}$ with dependence on $k$}
\end{table}

\subsection{Proof of the bound \eqref{main_bound}}

First of all, let us formulate some definitions.

\begin{definition}
Let $\PC = \PC(A,b)$ be a polyhedron as in definition of Theorem \ref{deep_base_th}. The set of indices $\BC \subseteq \intint m$ is a \emph{$\Delta$-deep base} if
\begin{enumerate}
    \item $|\BC| = n$ and $\det(A_{\BC}) \not= 0$;\\
    \item the following system is feasible:
    $$
    \begin{cases}
        b_{\BC} - (\Delta-1) \cdot \BUnit_n \leq A_{\BC}\, x \leq b_{\BC}\\
        A_{\NotBC}\, x \leq b_{\NotBC}\\
        x \in \RR^n,
    \end{cases}
    $$ where $\NotBC = \intint{m} \setminus \BC$.
\end{enumerate}


Let us denote the number of $\Delta$-deep bases of $\PC$ by  $\beta_{\Delta}(\PC)$.
\end{definition}

\begin{definition}
Let $\MC \subseteq \intint[0]{\Delta-1}^n$ be a convex-independent set, i.e. any point of $\MC$ can not be expressed as a convex combination of other points from $\MC$. 

Let us denote the maximal possible cardinally of $\MC$ by  $\gamma(n,\Delta)$.
\end{definition}

\begin{lemma}\label{vertex_bound_lm}
Let $\PC = \PC(A,b)$ be a polyhedron as in definition of Theorem \ref{deep_base_th}. Then, we have:
\begin{equation*}\label{inter_ineq}
    |\vertex(\PC_I)| \leq \beta_{\Delta}(\PC) \cdot \gamma(n,\Delta).
\end{equation*}

\end{lemma}
\begin{proof}
Let us consider the family $\BS$ of all possible $\Delta$-deep bases of $P$. For $\BC \in \mathscr{B}$, we use the following notation
$$
\VC_{\BC} = \{v \in \vertex(\PC_I) \colon b_{\BC} - A_{\BC} v < \Delta \cdot \BUnit\}.
$$

Due to Theorem \ref{deep_base_th}, we have
$
\vertex(\PC_I) = \bigcup_{\BC \in \BS} \VC_{\BC}
$. 

Now, we are going to estimate $|\VC_{\BC}|$. Let $\UC_{\BC} = \{ b_{\BC} - A_{\BC} v \colon v \in \VC_{\BC}\}$. Clearly, there exists a bijection between $\UC_{\BC}$ and $\VC_{\BC}$. 
Since $\VC_{\BC}$ is a convex-independent set, the same is true for $\UC_{\BC}$. Moreover $0 \leq u < \Delta \cdot \BUnit$, for $u \in U_{\BC}$. Consequently, $|\VC_{\BC}| = |\UC_{\BC}| \leq \gamma(n,\Delta)$, and $|\vertex(P_I)| \leq \beta_{\Delta}(\PC) \cdot \gamma(n,\Delta)$.
\end{proof}

\begin{corollary}\label{base_delta_cor}
In the assumptions of Theorem \ref{deep_base_th}, the following statements hold: 
\begin{enumerate}
    \item For any  $v \in \vertex(\PC_I)$, there exists a $\Delta$-deep base $\BC$ such that $\|b_{\BC} - A_{\BC} v\|_{\infty} \leq \Delta-1$;
    \item Additionally, $\abs{\supp_{\Delta-1}(b - A v)} \leq m-n$;
    \item The inequality $|\vertex(\PC_I)| \leq 2 \cdot \binom{m}{n} \cdot \Delta^{n-1}$ holds.
\end{enumerate}
\end{corollary}
\begin{proof}
Propositions 1 and 2 are straightforward consequences of Theorem \ref{deep_base_th}. Let us prove Proposition 3. Clearly, we have $\beta_{\Delta}(\PC) \leq \binom{m}{n}$. Due to Brass \cite{Brass1}, we have $\gamma(n,\Delta) \leq 2 \cdot \Delta^{n-1}$. Hence, the proof follows from Lemma \ref{vertex_bound_lm}.
\end{proof}

Using results of the papers \cite{columns_delta_modular} and  \cite{ModularDiffColumns} due to Averkov \& Schymura and Lee, Paat \& Stallknecht, we can give bounds that are independent on $m$:
\begin{corollary}
In the assumptions of Theorem \ref{deep_base_th}, the following statements hold:
\begin{enumerate}
    \item $\abs{\vertex(\PC_I)}  = O(n)^n \cdot \Delta^{3n-1}$;
    \item $\abs{\vertex(\PC_I)}  = O(n)^{3n} \cdot \Delta^{2n - 1}$.
\end{enumerate}
\end{corollary}
\begin{proof}
Clearly, $\binom{m}{n} = O\bigl(\frac{m}{n}\bigr)^n$. Due to \cite{columns_delta_modular} and \cite{ModularDiffColumns}, we can assume that $m = O(n^4 \cdot \Delta)$ or $m = O(n^2 \cdot \Delta^2)$ respectively.
\end{proof}

\subsection{Conclusions and directions for future research}

Due to Lemma \ref{vertex_bound_lm}, we estimate the integer vertices number of $\PC_I$ by \\$\beta_{\Delta}(\PC) \cdot \gamma(n,\Delta)$. Due to Erd{\H{o}}s, F{\"u}redi et al. \cite{grid_Erdos}, we have $\gamma(n,\Delta) \geq \frac{4}{n} \cdot \Delta^{n-2}$, so our bound \eqref{main_bound} on $|\vertex(\PC_I)|$ cannot be significantly improved, using only improvements on $\gamma(n,\Delta)$. On the other hand, we do not know any non-trivial upper or lower bounds for the number $\beta_{\Delta}(\PC)$ of $\Delta$-deep bases with respect to $\PC$. We believe that some significant improvements can be obtained, using accurate analysis of $\beta_{\Delta}(\PC)$.

\bibliographystyle{spmpsci}
\bibliography{grib_biblio}

\begin{thebibliography}{10}
\providecommand{\url}[1]{{#1}}
\providecommand{\urlprefix}{URL }
\expandafter\ifx\csname urlstyle\endcsname\relax
  \providecommand{\doi}[1]{DOI~\discretionary{}{}{}#1}\else
  \providecommand{\doi}{DOI~\discretionary{}{}{}\begingroup
  \urlstyle{rm}\Url}\fi

\bibitem{SupportIPSolutions}
Aliev, I., De~Loera, J.A., Eisenbrand, F., Oertel, T., Weismantel, R.: The
  support of integer optimal solutions.
\newblock SIAM Journal on Optimization \textbf{28}(3), 2152--2157 (2018).
\newblock \doi{10.1137/17M1162792}.
\newblock \urlprefix\url{https://doi.org/10.1137/17M1162792}

\bibitem{columns_delta_modular}
Averkov, G., Schymura, M.: On the maximal number of columns of a
  $\delta$-modular matrix.
\newblock arXiv preprint arXiv:2111.06294  (2021)

\bibitem{NewBoundsForFixedM}
Berndt, S., Jansen, K., Klein, K.M.: New Bounds for the Vertices of the Integer
  Hull, pp. 25--36.
\newblock \doi{10.1137/1.9781611976496.3}

\bibitem{Brass1}
Brass, P.: On lattice polyhedra and pseudocircle arrangements.
\newblock In: Karl der Grosse und sein Nachwirken. 1200 Jahre Kultur und
  Wissenschaft in Europa: Band II, Mathematisches Wissen, pp. 297--302 (1998)

\bibitem{IntVerticesSurveyPart2}
Chirkov A., Y., Veselov S., I.: On the vertices of implicitly defined integer
  polyhedra (part 2).
\newblock Vestnik of Lobachevsky University of Nizhni Novgorod (2), 166--172
  (2008).
\newblock In Russian

\bibitem{Sensitivity_Tardos}
Cook, W., Gerards, A.M.H., Schrijver, A., Tardos, E.: Sensitivity theorems in
  integer linear programming.
\newblock Mathematical Programming \textbf{34}(3), 251--261 (1986).
\newblock \doi{10.1007/BF01582230}

\bibitem{IntVert_Cook}
Cook, W., Hartmann, M., Kannan, R., McDiarmid, C.: On integer points in
  polyhedra.
\newblock Combinatorica \textbf{12}(1), 27--37 (1992).
\newblock \doi{10.1007/BF01191202}.
\newblock \urlprefix\url{https://doi.org/10.1007/BF01191202}

\bibitem{grid_Erdos}
Erd{\H{o}}s, P., F{\"u}redi, Z., Pach, J., Ruzsa, I.Z.: The grid revisted.
\newblock Discrete mathematics \textbf{111}(1-3), 189--196 (1993)

\bibitem{FPT_Grib}
Gribanov D., V., Malyshev D., S., Pardalos P., M., Veselov S., I.:
  {FPT}-algorithms for some problems related to integer programming.
\newblock J. Comb. Optim. \textbf{35}, 1128--1146 (2018).
\newblock \doi{10.1007/s10878-018-0264-z}.
\newblock \urlprefix\url{https://doi.org/10.1007/s10878-018-0264-z}

\bibitem{OnCanonicalProblems_Grib}
Gribanov D., V., Shumilov I., A., Malyshev D., S., Pardalos P., M.: On
  $\delta$-modular integer linear problems in the canonical form and equivalent
  problems (2021).
\newblock \urlprefix\url{https://arxiv.org/abs/2002.01307v5}

\bibitem{Grunbaum}
Gr{\"u}nbaum, B.: Convex Polytopes.
\newblock Graduate Texts in Mathematics. Springer-Verlag, New York (2011)

\bibitem{ModularDiffColumns}
Lee, J., Paat, J., Stallknecht, I., Xu, L.: Polynomial upper bounds on the
  number of differing columns of an integer program.
\newblock arXiv preprint arXiv:2105.08160v2 [math.OC]  (2021).
\newblock \urlprefix\url{https://arxiv.org/abs/2105.08160}

\bibitem{MaxFacesTh}
McMullen, P.: The maximum numbers of faces of a convex polytope.
\newblock Mathematika \textbf{17}(2), 179--184 (1970).
\newblock \doi{10.1112/S0025579300002850}

\bibitem{Schrijver}
Schrijver, A.: Theory of linear and integer programming.
\newblock John Wiley \& Sons, Chichester (1998)

\bibitem{IntVerticesSurveyPart1}
Veselov S., I., Chirkov A., Y.: On the vertices of implicitly defined integer
  polyhedra.
\newblock Vestnik of Lobachevsky University of Nizhni Novgorod (1), 118--123
  (2008).
\newblock In Russian

\bibitem{IntVertEstimates_VesChir}
Veselov S., I., Chirkov A., Y.: Some estimates for the number of vertices of
  integer polyhedra.
\newblock J. Appl. Ind. Math. \textbf{2}, 591--604 (2008).
\newblock \doi{10.1134/S1990478908040157}.
\newblock \urlprefix\url{https://doi.org/10.1134/S1990478908040157}

\end{thebibliography}

\section*{Statements and Declarations}

\subsection*{\bf Funding}
This work was prepared under financial support of Russian Science Foundation grant No 21-11-00194.

\subsection*{\bf Competing Interests}
The authors have no relevant financial or non-financial interests to disclose.

\subsection*{\bf Author Contributions}
All the authors contributed equally.

\subsection*{\bf Data Availability}
The manuscript has no associated data.

\end{document}